\documentclass[aps,pra,reprint,superscriptaddress,amsmath,amssymb,longbibliography]{revtex4-2}
\usepackage[utf8]{inputenc}
\usepackage[T1]{fontenc}
\usepackage{graphicx}
\usepackage{tikz}
\usepackage{xcolor}
\usepackage{colortbl}
\usepackage{amsthm}
\usepackage{times}
\newcommand{\ket}[1]{\lvert #1\rangle}
\newcommand{\braket}[2]{\langle #1\vert #2\rangle}
\usepackage[colorlinks=true,citecolor=blue,urlcolor=blue]{hyperref}
\theoremstyle{plain}
\newtheorem{proposition}{Proposition}
\newtheorem{lemma}{Lemma}
\newtheorem{theorem}{Theorem}

\begin{document}
\raggedbottom
\title{Euler's thirty-six officers need not all be entangled:\\ real \texorpdfstring{AME$(4,6)$}{AME(4,6)}, families and classification}
\author{Palash Pandya}
\affiliation{Center for Theoretical Physics, Polish Academy of Sciences, Aleja Lotnik\'{o}w 32/46, 02-668 Warsaw, Poland}
\affiliation{Center for Quantum-Enabled Computing, Center for Theoretical Physics, Polish Academy of Sciences, Aleja Lotnik\'{o}w 32/46, 02-668 Warsaw, Poland}

\begin{abstract}
Absolutely maximally entangled states of four six-level systems, AME$(4,6)$, are the quantum solutions of Euler's
problem of the thirty-six officers.  Until recently all reported solutions were complex, and it was conjectured that no
real orthogonal 2-unitary matrix of order $36$ exists.  We give real AME$(4,6)$ states with $60$ non-zero
coefficients in $\{\pm1,\sqrt2\}/\sqrt{72}$.  Of their thirty-six officers, $12$ are classical, as some must be whenever
fewer than $72$ coefficients are non-zero.  They lie in the $-1$
eigenspace of an order-$4$ local symmetry, and on their support the AME conditions become linear.  We derive two closed-form families of AME$(4,6)$ states that meet at one of these states: a $10$-parameter family that
is supported on $72$ basis states, and an $11$-parameter family on $60$ basis states.  At each of its two ends the first family loses $12$ coefficients, and the AME states on the remaining $60$ basis states
form an $11$-parameter family: the second family at one end, and an analogous family at the other.  These two $60$-term families contain $4096$ real AME states up to diagonal local unitaries, $1024$ of which are limits of the first family.  They fall into exactly $20$ classes under local unitaries and permutations of the parties, and we certify their separation with interval arithmetic.
\end{abstract}
\maketitle

\section{Introduction}

A pure state of $N$ parties with $d$ levels each is absolutely maximally entangled, AME$(N,d)$, if every reduction
to $\lfloor N/2\rfloor$ parties is maximally mixed~\cite{Helwig2012,GALRZ15}.  Such states are perfect tensors.  They
underlie quantum secret sharing~\cite{Helwig2012} and holographic codes~\cite{PYHP2015}, and 2-unitary matrices
are dual-unitary gates of maximal entangling power~\cite{RAL20,BKP2019}.  For four parties, writing $U_{(ij),(kl)}=d\,\psi_{ijkl}$, the state is AME if and only if
$U$ is \emph{2-unitary}: $U$, its realignment $U^R$ and its partial transpose $U^\Gamma$ are all
unitary~\cite{GALRZ15}.  Up to permutations of rows and columns, these three matrices are the reshapings of
$d\,\psi$ along the three bipartitions $12|34$, $13|24$ and $14|23$.

A pair of orthogonal Latin squares of order $d$ gives a 2-unitary permutation matrix, the classical case of the quantum
combinatorial designs of Refs.~\cite{GALRZ15,GRMZ2018}.  Such pairs exist for every
$d\ge3$ except $d=6$~\cite{BSP1960}: Euler's problem of the $36$ officers has no
solution~\cite{Euler1782,Tarry1900}.  AME$(4,2)$ states do not exist at all~\cite{HS2000}.  An AME$(4,6)$ state, a
quantum solution of Euler's problem, was found by Rather \emph{et al.}~\cite{RBB22}.  Further solutions include
continuous families obtained from it~\cite{RRKL23}, 2-unitary complex Hadamard matrices~\cite{BZ24}, constructions
from biunimodular vectors~\cite{Rat24} and hand-made solutions~\cite{GG25}.  No pure stabilizer AME$(4,6)$ state
exists~\cite{WMBA26,Cha26}, and every quantum solution requires entanglement: mutually orthogonal quantum Latin
squares of order six without entanglement do not exist~\cite{BS26}.

All of these solutions are given with complex coefficients.  In the Supplemental Material of Ref.~\cite{RBB22} the authors write:
``Further numerical results suggest that any 2-unitary matrix $U_{36}$ has complex entries.  Therefore, it is
tempting to conjecture that there exists no solution in the set $O(36)$ of orthogonal matrices of this size.''  A
recent review lists the question as open problem (T2)~\cite{RBRLZ26}: ``Is there an AME(4,6) state with real
coefficients? Equivalently, is there an orthogonal 2-unitary matrix of order 36?''
This question has now been answered in the affirmative, independently, by Bruzda~\cite{Bru26},
Rather~\cite{Rat26}, Hontarenko and \.Zyczkowski~\cite{HZ26}, and the present work.  With the permutation solutions for $d\neq2,6$, it follows that a real
AME$(4,d)$ state exists if and only if $d\ge3$, exactly when a complex one does.  We present the following results in this work.
Section~\ref{sec:state} presents our central result: an explicit real AME$(4,6)$ state with $60$ non-zero coefficients, the sparsest AME$(4,6)$ state we know of,
whose 2-unitary representation consists of direct sums of $1\times1$ and $2\times2$ blocks.  Of its $36$ officers, $12$ are
classical, as some must be in any AME$(4,6)$ state with fewer than $72$ non-zero coefficients (Sec.~\ref{sec:discussion}).
Section~\ref{sec:symmetry} describes its origin.  It lies in the $-1$ eigenspace of the local symmetry
$u^{\otimes4}$, $u=\mathrm{diag}(i,-1,-1,-i,-i,-i)$, and we show that the AME conditions are linear on its support.
In Sec.~\ref{sec:families} two continuous families of AME$(4,6)$ states are presented in closed form, with supports on $72$ and $60$ basis states.  The two families meet at the real state.  In Sec.~\ref{sec:real} we show that the families allow a complete classification of the real AME states with the two $60$-term supports at the ends of the $72$-term family: $4096$ states in exactly $20$ classes under local unitaries and party permutations, and we prove that each of them is isolated among real AME states modulo local orthogonal maps.
All statements about our states are verified by self-contained scripts~\cite{repo}.

\section{A real \texorpdfstring{AME$(4,6)$}{AME(4,6)} state}\label{sec:state}

If we index the basis states of each party as $\ket{i}$, where $i = 0,\dots,5$, then we can write the state as an array of Euler's officers where the first two
parties carry rank $a$ and regiment $b$, and the last two parties carry the state $\ket{\phi_{ab}}$ attached to them:
\begin{equation}\label{eq:state}
\ket\psi=\frac{1}{\sqrt{72}}\sum_{a,b=0}^{5}\ket{ab}\otimes\ket{\phi_{ab}}\,.
\end{equation}

\begin{table*}[t]
\caption{The real AME$(4,6)$ state of Eq.~\eqref{eq:state} as Euler's array of officers: row $a$ (rank), column $b$
(regiment), entry $\ket{\phi_{ab}}$, the unnormalized state of the third and fourth party.  Shaded: the $12$ classical
officers $\sqrt2\,\ket{cd}$.  All other entries are two-term states $\pm\ket{cd}\pm\ket{c'd'}$.}
\label{tab:state}
\begin{center}
\renewcommand{\arraystretch}{1.4}\setlength{\tabcolsep}{8pt}
\begin{tabular}{c|cccccc}
$\ket{\phi_{ab}}$ & $b=0$ & $b=1$ & $b=2$ & $b=3$ & $b=4$ & $b=5$\\\hline
$a=0$ & $-\ket{12}+\ket{21}$ & $-\ket{02}+\ket{20}$ & $-\ket{01}+\ket{10}$ & \cellcolor{gray!18}$\sqrt2\,\ket{44}$ & \cellcolor{gray!18}$\sqrt2\,\ket{55}$ & \cellcolor{gray!18}$\sqrt2\,\ket{33}$\\
$a=1$ & $-\ket{01}-\ket{10}$ & $\ket{00}+\ket{54}$ & $\ket{35}+\ket{43}$ & $-\ket{15}-\ket{32}$ & $-\ket{23}+\ket{41}$ & $-\ket{24}+\ket{52}$\\
$a=2$ & $\ket{02}+\ket{20}$ & $-\ket{35}+\ket{43}$ & $\ket{00}-\ket{54}$ & $-\ket{25}+\ket{31}$ & $-\ket{13}+\ket{42}$ & $\ket{14}+\ket{51}$\\
$a=3$ & \cellcolor{gray!18}$\sqrt2\,\ket{34}$ & $-\ket{23}-\ket{41}$ & $\ket{13}+\ket{42}$ & \cellcolor{gray!18}$\sqrt2\,\ket{50}$ & $-\ket{11}+\ket{22}$ & \cellcolor{gray!18}$\sqrt2\,\ket{05}$\\
$a=4$ & \cellcolor{gray!18}$\sqrt2\,\ket{53}$ & $\ket{15}-\ket{32}$ & $\ket{25}+\ket{31}$ & $\ket{11}+\ket{22}$ & \cellcolor{gray!18}$\sqrt2\,\ket{04}$ & \cellcolor{gray!18}$\sqrt2\,\ket{40}$\\
$a=5$ & \cellcolor{gray!18}$\sqrt2\,\ket{45}$ & $-\ket{14}+\ket{51}$ & $\ket{24}+\ket{52}$ & \cellcolor{gray!18}$\sqrt2\,\ket{03}$ & \cellcolor{gray!18}$\sqrt2\,\ket{30}$ & $\ket{12}+\ket{21}$\\
\end{tabular}
\end{center}
\end{table*}
The $36$ states $\ket{\phi_{ab}}$, for one officer each, are listed in Table~\ref{tab:state}.
The $12$ shaded cells are \emph{classical} officers: a single
basis state with coefficient $\sqrt2/\sqrt{72}=1/6$.  The other $24$ are \emph{entangled} officers, each a
superposition of two basis states with coefficients $\pm1/\sqrt{72}$.  The support of $\ket\psi$ thus has
$12+2\cdot24=60$ basis states $\ket x=\ket{x_1x_2x_3x_4}$, where $x_j\in\{0,\dots,5\}$ is the index of party $j$.
With $\psi_x=\braket{x}{\psi}$ we write
\[
P=\{x:\psi_x=1/6\},\qquad E=\{x:\psi_x=\pm1/\sqrt{72}\}
\]
for the $12$ basis states of the classical officers ($P$) and the $48$ of the entangled ones ($E$).
For $x\in P$ all four indices $x_j$ lie in $\{0,3,4,5\}$ and exactly one of them is $0$, e.g.\ $\ket{0344}$.
For $x\in E$ exactly two of them lie in $\{1,2\}$, e.g.\ $\ket{0012}$.

\begin{proposition}\label{prop:state}
$\ket\psi$ is an AME$(4,6)$ state: the reshaped matrix $U_{(ij),(kl)}=6\psi_{ijkl}$ is a real orthogonal 2-unitary matrix of order $36$.
\end{proposition}

\begin{proof}
We first show that $U$ is orthogonal, i.e.\ that the bipartition $12|34$ is maximally entangled:
$\rho_{12}=UU^T/36=\mathbb I/36$.  The row $(a,b)$ of $U$ is the officer $\ket{\phi_{ab}}/\sqrt2$ of
Table~\ref{tab:state}, written in the basis $\ket{x_3x_4}$ of parties $3,4$.  A classical officer gives a row with one
entry $1$, an entangled officer a row with two entries $\pm1/\sqrt2$.  Every ket $\ket{x_3x_4}$ appears either in
exactly one cell, or in exactly two cells that contain the same two kets.  The rows of two such cells are therefore
supported on the same two columns, which no other row uses, and after permuting rows and columns $U$ is a direct sum of
$12$ blocks $(1)$ and $12$ blocks $\frac1{\sqrt2}\bigl(\begin{smallmatrix}s_1&s_2\\ s_3&s_4\end{smallmatrix}\bigr)$
with $s_i=\pm1$.  For example, $(a,b)=(0,0)$ and $(5,5)$ give
$\frac1{\sqrt2}\bigl(\begin{smallmatrix}-1&1\\ 1&1\end{smallmatrix}\bigr)$.  The two officers of a pair are
orthogonal, $s_1s_3+s_2s_4=0$, so every block, and hence $U$, is orthogonal.

The same argument applies to $U^R$ and $U^\Gamma$, the reshapings along $13|24$ and $14|23$, whose rows collect the
basis states $\ket x$ with given indices $(x_1,x_3)$ or $(x_1,x_4)$.  Each of them again has $12$ rows with one entry
$1$ and $24$ rows with two entries $\pm1/\sqrt2$, paired on the same two columns with orthogonal signs, so both are
direct sums of $12$ blocks $(1)$ and $12$ orthogonal $2\times2$ blocks, whose block lists are part of the exact check
in Ref.~\cite{repo}.  Hence also $\rho_{13}=\rho_{14}=\mathbb I/36$, and the complementary pairs follow because
$\ket\psi$ is pure.  So $U$ is 2-unitary and $\ket\psi$ is AME.
\end{proof}

All one- and two-party marginals have also been verified in exact arithmetic over $\mathbb Q(\sqrt2)$~\cite{repo}.
Fig.~\ref{fig:amplitudes} shows the three reshapings.  In the golden state all
$36$ officers are entangled~\cite{RBB22}.  The classical officers also make $\ket\psi$ as close to a product
state as an AME$(4,6)$ state can be: for every product state $\ket v=\ket{v_1v_2v_3v_4}$,
$|\braket{v}{\psi}|^2\le\langle v_1v_2|\rho_{12}|v_1v_2\rangle=1/36$, so the overlap $|\braket{v}{\psi}|$ is at most $1/6$, with
equality at the $12$ states of $P$.

\begin{figure*}[t]
\centering
\includegraphics[width=\textwidth]{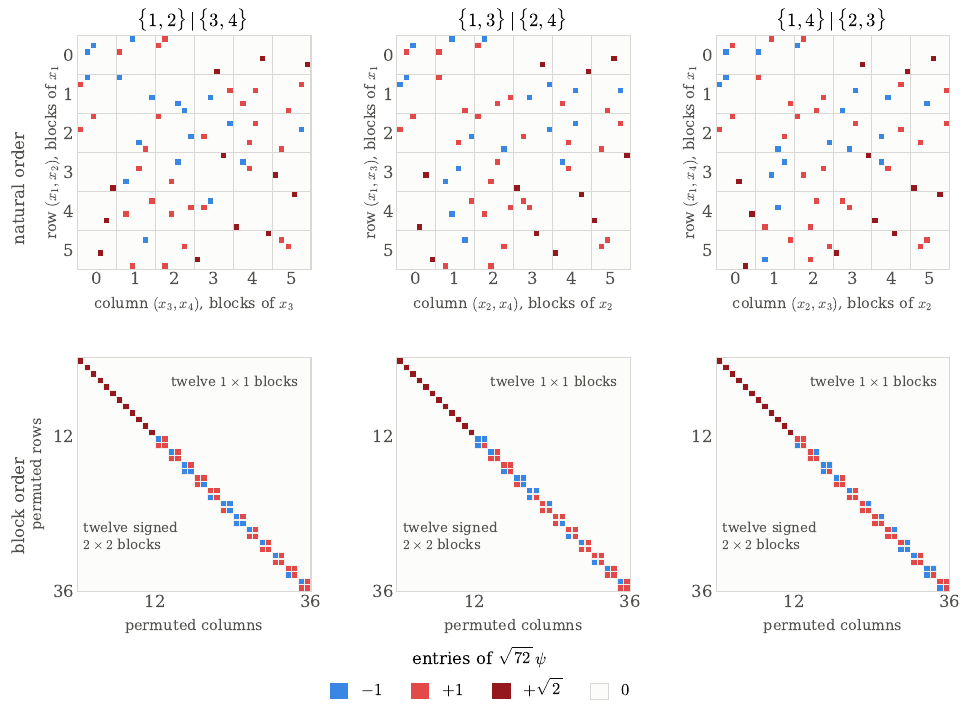}
\caption{The three reshapings of $\sqrt{72}\,\ket\psi$, the real AME$(4,6)$ state, for the bipartitions $12|34$, $13|24$ and $14|23$.  Top: natural
order.  Bottom: rows and columns permuted so that the $12$ entries $\sqrt2$ form $1\times1$ blocks and the $48$
entries $\pm1$ form $12$ blocks of size $2\times2$, each proportional to an orthogonal matrix.}
\label{fig:amplitudes}
\end{figure*}

\section{Symmetry and linearization}\label{sec:symmetry}

The state was obtained as a byproduct of a search
for AME states in eigenspaces
of local symmetries $u^{\otimes4}$, $u=VDV^\dagger$, with $V$ Haar random and $D$ diagonal of finite order.  Two
necessary conditions restrict $D$ and the eigenvalue $\lambda$.  First, if $u^{\otimes4}\ket\Phi=\lambda\ket\Phi$, the
reshaping of $\ket\Phi$ along a cut $A|A^c$ connects the $\mu$-eigenspace of $u^{\otimes|A|}$ only with the
$\lambda\bar\mu$-eigenspace of $u^{\otimes|A^c|}$.  For $|A|=2$ the reshaping is proportional to a unitary, so these
eigenspaces must have equal dimensions for every $\mu$, and for $|A|=1$ the second must be at least as large.
Second, in the eigenbasis of $u$ the populations $|\Phi_x|^2$ must form a probability distribution on the basis states of the
eigenspace with uniform two-party marginals (these are the diagonals of $\rho_{jk}=\mathbb I/36$), which is formulated as a linear feasibility problem.
In particular, with complex unit $i$, take $V=\mathbb I$ and $u=D=\mathrm{diag}(i^{k_0},\dots,i^{k_5})$ with $k=(1,2,2,3,3,3)$.
The $-1$ eigenspace of $u^{\otimes4}$ is spanned by the $360$ basis states $\ket{x_1x_2x_3x_4}$ with $k_{x_1}+\dots+k_{x_4}\equiv2\pmod4$, and a random-start optimization
in it finds AME states within a few restarts.  Most of them have a $10$-dimensional space of first-order deformations
modulo local unitaries (Sec.~\ref{sec:families}), the others a $3$-dimensional one.\footnote{A first-order deformation
of an AME state $\ket\Phi$ is a vector $\ket{\delta\Phi}$ such that the two-party marginals of $\ket\Phi+t\ket{\delta\Phi}$
remain $\mathbb I/36$ to first order in $t$.  For four parties it suffices to check $\rho_{12}$, $\rho_{13}$ and
$\rho_{14}$.  These vectors form a real vector space.  It always contains the tangent space of the local-unitary orbit,
spanned by the vectors $i(H_1\otimes\mathbb I\otimes\mathbb I\otimes\mathbb I+\dots+\mathbb I\otimes\mathbb I\otimes\mathbb
I\otimes H_4)\ket\Phi$ with Hermitian $H_j$, and the dimension of the quotient is invariant under local unitaries.  It
is an upper bound on the number of parameters, modulo local unitaries, of any family of AME states through $\ket\Phi$,
with equality where the rank of the linearized conditions is locally constant.}
Local unitaries that commute with $u$, namely $U(1)\times U(2)\times U(3)$ acting on the basis states
$\ket0$; $\ket1,\ket2$; and $\ket3,\ket4,\ket5$ of each party, preserve the eigenspace.  Minimizing a smoothed $\ell_1$
norm of the coefficients over such changes of local basis can bring a solution with a $10$-dimensional deformation space to exactly $72$
non-zero coefficients. A subsequent Gauss--Newton optimization on that support makes it AME to machine precision~\cite{repo}.  These $72$ basis states form an orthogonal array of strength $2$ and index
$2$~\cite{HSS1999}.
They consist of the $60$ basis states $E\cup P$ of the support of $\ket\psi$ and $12$ further ones, which we call $P'$
(listed in Eq.~\eqref{eq:quartets}).  As for $P$, all four indices of each $x\in P'$ lie in $\{0,3,4,5\}$ and exactly one of them is
$0$.  Altogether, the $72$ basis states lie in the $-1$ eigenspace ($\sum_jk_{x_j}\equiv2\pmod4$ for each), so every state supported on them, including
$\ket\psi$ and both families of Sec.~\ref{sec:families}, is a $-1$ eigenvector of $u^{\otimes4}$.  Modulo an overall
phase on each party, the diagonal local unitaries that act on this support as a scalar form the group
$\mathbb Z_2\times\mathbb Z_2\times\mathbb Z_4$, in which $u^{\otimes4}$ has order $4$.

An irredundant orthogonal array of index $1$ gives an AME state as the uniform superposition of its rows~\cite{GZ2014}.
With index $2$ the phases must also cancel.  In each of the three reshapings, the rows and the columns are labeled by the
indices of two complementary pairs of parties, so strength $2$ and index $2$ mean that every row and every column has exactly two non-zero entries.
Moreover, two rows that share a column share both of their columns, so after permuting rows and columns each reshaping
is a direct sum of $18$ full $2\times2$ blocks~\cite{repo}: the $12$ blocks of $E$ (the \emph{$E$-blocks}) from the
proof of Proposition~\ref{prop:state}, and six \emph{mixed} blocks with two states of $P$ on one diagonal and two of
$P'$ on the other.  Over the three reshapings this gives $36$ $E$-blocks and $18$ mixed blocks.  A
block-diagonal matrix is unitary if and only if each of its blocks is, and for a $2\times2$ block this amounts to
linear equations in the squared moduli and a single linear congruence in the phases: on such a support the AME
conditions linearize.  On the support $E\cup P$ of $\ket\psi$ the six mixed blocks split into $12$ blocks of size
$1\times1$, and the proof of Proposition~\ref{prop:state} is this argument.  Allowing both kinds of blocks gives the
following lemma.

\begin{lemma}\label{lem:blocks}
Let $\ket\Phi=\sum_{x\in S}\Phi_x\ket x$ with $\Phi_x=\sqrt{s_x/72}\,e^{i\varphi_x}$, where $s_x>0$ and
$\varphi_x\in\mathbb R$.  In each of the three reshapings $U$, $U^R$, $U^\Gamma$ of $6\ket\Phi$, along $12|34$,
$13|24$ and $14|23$, the basis state $\ket x$ lies in the row labeled by its indices for the first pair of parties and
in the column labeled by those for the second.  For $x\in S$ in row $r$ and column $c$ we write $s_{rc}=s_x$ and
$\varphi_{rc}=\varphi_x$.  Suppose that in each reshaping the $36$ rows and the $36$ columns can be partitioned into
sets $R_1,\dots,R_m$ and $C_1,\dots,C_m$ with $|R_k|=|C_k|\in\{1,2\}$ such that $S$ consists exactly of the basis
states in $R_1\times C_1,\dots,R_m\times C_m$, the \emph{blocks} of $S$.  Then $\ket\Phi$ is AME if and only if
$s_x=2$ on every $1\times1$ block and every $2\times2$ block, with rows $r,r'$ and columns $c,c'$, satisfies
\begin{equation}\label{eq:congruence}
\begin{gathered}
s_{rc}=s_{r'c'},\quad s_{rc'}=s_{r'c},\quad s_{rc}+s_{rc'}=2,\\
\varphi_{rc}-\varphi_{rc'}-\varphi_{r'c}+\varphi_{r'c'}\equiv\pi\pmod{2\pi}.
\end{gathered}
\end{equation}
In particular, if these moduli equations have a solution with all $s_x>0$, a real AME state with support $S$ exists
if and only if the equations $t_{rc}+t_{rc'}+t_{r'c}+t_{r'c'}=1$, one for each $2\times2$ block, have a solution
$t\in\mathbb F_2^S$.
\end{lemma}

\begin{proof}
The reduced state of parties $1,2$ is $\rho_{12}=UU^\dagger/36$, and that of parties $3,4$ is
$\rho_{34}=\overline{U^\dagger U}/36$, and $U^R$ and $U^\Gamma$ give $\rho_{13},\rho_{24}$ and $\rho_{14},\rho_{23}$ in
the same way.  Hence $\ket\Phi$ is AME, i.e.\ all six
two-party reductions equal $\mathbb I/36$ (which includes $\braket{\Phi}{\Phi}=1$), if and only if $U$, $U^R$ and
$U^\Gamma$ are unitary.

Fix one reshaping $U$ and its blocks.  The entry of $U$ in row $r$ and column $c$ is $\sqrt{s_{rc}/2}\,e^{i\varphi_{rc}}$
if the basis state there lies in $S$, and $0$ otherwise.  Since the $R_k$ and the $C_k$ partition the rows and the
columns and all entries of $S$ lie in the blocks $R_k\times C_k$, ordering the rows as $R_1,\dots,R_m$ and the columns
as $C_1,\dots,C_m$ makes $U$ block diagonal, with diagonal blocks $B_k$ the submatrices on $R_k\times C_k$.  So $U$ is unitary if and only if every $B_k$ is.  A $1\times1$ block is unitary if and
only if $s_x/2=1$.  A $2\times2$ block $B$ with rows $r,r'$ and columns $c,c'$ is unitary if and only if its rows are
orthonormal:
\[
s_{rc}+s_{rc'}=s_{r'c}+s_{r'c'}=2,\qquad B_{rc}\bar B_{r'c}+B_{rc'}\bar B_{r'c'}=0.
\]
Then $B^\dagger B=\mathbb I$ as well, so also $s_{rc}+s_{r'c}=s_{rc'}+s_{r'c'}=2$.  These four norm equations are
equivalent to the three moduli equations in~\eqref{eq:congruence}.  Under them, the two terms of the orthogonality
condition have the same modulus $\sqrt{s_{rc}s_{rc'}}/2>0$, so their sum vanishes if and only if their phases
$\varphi_{rc}-\varphi_{r'c}$ and $\varphi_{rc'}-\varphi_{r'c'}$ differ by $\pi$ modulo $2\pi$, which is the congruence.
Conversely, the moduli equations and the congruence give orthonormal rows.  So $B$ is unitary if and only if
Eq.~\eqref{eq:congruence} holds.  Exchanging $r$ with $r'$, or $c$ with $c'$, gives an equivalent system, so the
labeling within a block does not matter.  Collecting the conditions of all blocks of the three reshapings proves the
first claim.

The moduli equations involve only $s$ and the congruences only $\varphi$, so any positive solution $s$ and any
solution $\varphi$ together give an AME state with support $S$, and every such state arises in this way.  It is real
if and only if all $\varphi_x\in\{0,\pi\}$.  Writing $\varphi_x=\pi t_x$ with $t_x\in\{0,1\}$, the congruence becomes
$t_{rc}-t_{rc'}-t_{r'c}+t_{r'c'}\equiv1\pmod2$, which is the stated equation over $\mathbb F_2$.
\end{proof}

\section{Two families through the real state}\label{sec:families}

We now describe, up to diagonal local unitaries, all AME states with support exactly $E\cup P\cup P'$, $E\cup P$ or
$E\cup P'$.  Let $\varepsilon_x=\pm1$ be the sign of $\psi_x$ for $x\in E$.  For $\delta\in[-1,1]$ and phases
$\alpha_x$ ($x\in P$) and $\beta_x$ ($x\in P'$) let
\begin{multline}\label{eq:family}
\ket{\Psi_\delta(\alpha,\beta)}=\frac1{\sqrt{72}}\Big[\sum_{x\in E}\varepsilon_x\ket x\\
+\sqrt{1+\delta}\sum_{x\in P}e^{i\alpha_x}\ket x+\sqrt{1-\delta}\sum_{x\in P'}e^{i\beta_x}\ket x\Big].
\end{multline}
At $\delta=1$ the coefficients on $P'$ vanish and we write $\ket{\Psi_1(\alpha)}$, since the phases $\beta$ play no
role.  At $\delta=-1$ those on $P$ vanish, and we write $\ket{\Psi_{-1}(\beta)}$.  A number in place of
$\alpha$ or $\beta$ means that all phases on $P$, or all phases on $P'$, have that value.  Thus, the real state in Eq.~\eqref{eq:state} is $\ket\psi=\ket{\Psi_1(0)}$,
while the other end with all phases zero is a different real AME state,
\begin{equation}\label{eq:psiminus}
\ket{\psi_-}=\ket{\Psi_{-1}(0)}=\frac1{\sqrt{72}}\sum_{x\in E}\varepsilon_x\ket x+\frac16\sum_{x\in P'}\ket x ,
\end{equation}
with the $12$ coefficients $1/6$ on $P'$ instead of $P$.  No local unitary and permutation of the parties maps
$\ket\psi$ to $\ket{\psi_-}$ (Theorem~\ref{thm:classes}), and no permutation of the parties and of the basis states
of each party maps the support $E\cup P$ onto $E\cup P'$.  Each of the
$18$ mixed blocks has two states $x,y\in P$ on one diagonal and two states $x',y'\in P'$ on
the other.  For $|\delta|<1$ the moduli of Eq.~\eqref{eq:family} satisfy the moduli equations of
Lemma~\ref{lem:blocks} on every block, and the $E$-blocks carry the coefficients of $\ket\psi$ which satisfy their respective
phase congruences.  The only remaining AME conditions are therefore the phase congruences~\eqref{eq:congruence} of the
mixed blocks, which read
\begin{equation}\label{eq:mixed}
\alpha_x+\alpha_y\equiv\beta_{x'}+\beta_{y'}+\pi\pmod{2\pi}.
\end{equation}
The $18$ pairs $\{x,y\}$ are the pairs of states within three quartets $P_1,P_2,P_3$ that partition $P$, the $18$
pairs $\{x',y'\}$ are those within three quartets $P'_1,P'_2,P'_3$ that partition $P'$, and each pair occurs in exactly
one mixed block:
\begin{equation}\label{eq:quartets}
\begin{aligned}
P_1&=\{0344,3505,4053,5430\},\\
P_2&=\{0455,3034,4540,5303\},\\
P_3&=\{0533,3350,4404,5045\},\\
P'_1&=\{0353,3530,4044,5405\},\\
P'_2&=\{0434,3055,4503,5340\},\\
P'_3&=\{0545,3304,4450,5033\}.
\end{aligned}
\end{equation}
The mixed blocks join the quartets in a regular way: for all $k$ and $l$, exactly two of them contain states of $P_k$
and $P'_l$, and their pairs split $P_k$ into two pairs $p,\bar p$ and $P'_l$ into two pairs $q,\bar q$, where partner
pairs have the same first indices~\cite{repo}.  With $\alpha(p)$ the sum of $\alpha$ over $p$, Eq.~\eqref{eq:mixed}
for these two blocks reads
\begin{equation}\label{eq:partition}
\alpha(p)\equiv\beta(q)+\pi,\qquad\alpha(\bar p)\equiv\beta(\bar q)+\pi.
\end{equation}
We call the AME states with support exactly $E\cup P$ or $E\cup P'$ the two \emph{faces},
$\mathcal F_{60}$ and $\mathcal F_{60}'$.  Eq.~\eqref{eq:family} covers both faces at $\delta=\pm1$, but a face is more
than the set of limits of the states with $|\delta|<1$: by Proposition~\ref{prop:families}(c), the AME states with
support $E\cup P\cup P'$ reach each face only along a $9$-dimensional subtorus.

\begin{figure*}[t]
\centering
\includegraphics[width=\textwidth]{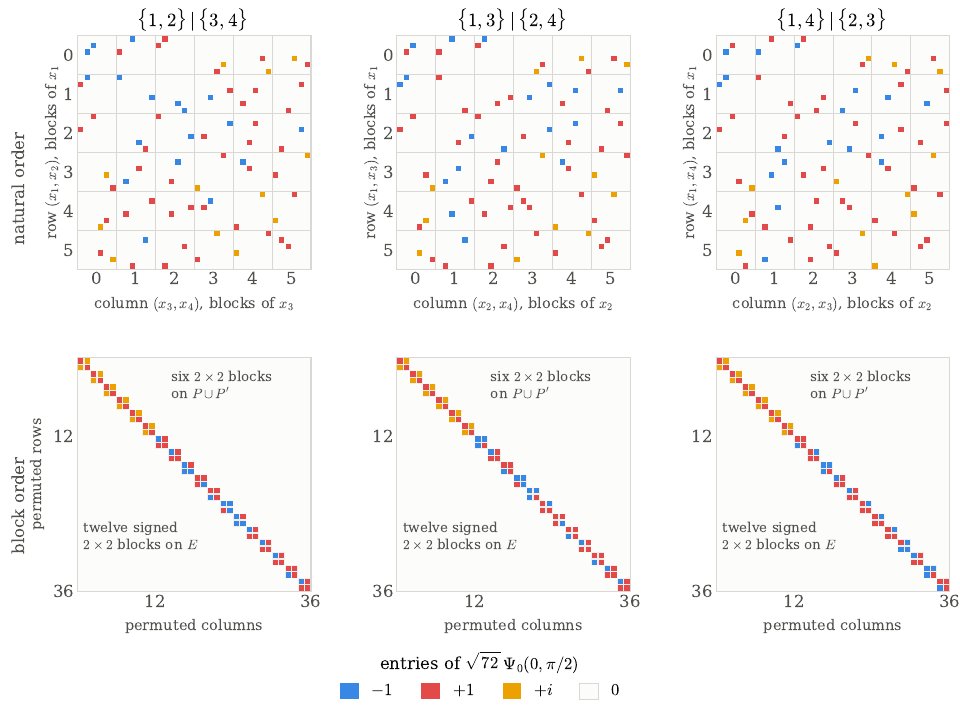}
\caption{The three reshapings of $\sqrt{72}\,\ket{\Psi_0(0,\pi/2)}$, a uniform-amplitude state of $\mathcal F_{72}^+$, for the
bipartitions $12|34$, $13|24$ and $14|23$.  Top: natural order.  Bottom: rows and columns permuted into $18$ blocks
of size $2\times2$: six with the entries $+1$ of $P$ on the diagonal and $+i$ of $P'$ off it, then the $12$ blocks
of $E$, ordered as in Fig.~\ref{fig:amplitudes}.  As $\delta\to1$ the $P'$ entries vanish, the $P$ entries grow to
$\sqrt2$, and the picture becomes Fig.~\ref{fig:amplitudes}.}
\label{fig:amplitudes72}
\end{figure*}

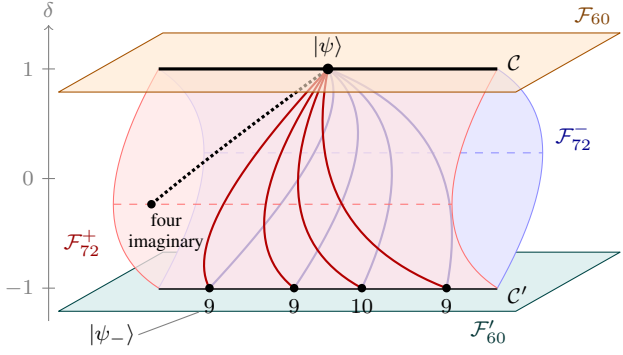
\begin{figure}[tb]
\centering
\begin{tikzpicture}[x={(0.56cm,0cm)},y={(0.30cm,0.17cm)},z={(0cm,1.45cm)},line join=round,font=\footnotesize]
  \def\b{2.0}
  \def\ends{1.2/9,3.2/9,4.8/10,6.8/9}
  \draw[->,gray] (-2.6,0,-1.3) -- (-2.6,0,1.4) node[above] {$\delta$};
  \foreach \z in {1,0,-1} { \draw[gray] (-2.75,0,\z) -- (-2.45,0,\z) node[pos=0,left] {$\z$}; }
  \fill[teal!12] (-1.4,-1.8,-1) -- (9.4,-1.8,-1) -- (9.4,2.8,-1) -- (-1.4,2.8,-1) -- cycle;
  \draw[teal!55!black] (-1.4,-1.8,-1) -- (9.4,-1.8,-1) -- (9.4,2.8,-1) -- (-1.4,2.8,-1) -- cycle;
  \node[teal!40!black,anchor=north east] at (9.4,-1.8,-1) {$\mathcal F_{60}'$};
  \draw[very thick] (0,0,-1) -- (8,0,-1);
  \node[anchor=west] at (8.05,0,-1.06) {$\mathcal C'$};
  \fill[blue!9] plot[domain=1:-1,variable=\t,samples=31] (8,{\b*(1-\t*\t)},\t)
                -- plot[domain=-1:1,variable=\t,samples=31] (0,{\b*(1-\t*\t)},\t) -- cycle;
  \foreach \xx in {0,8} { \draw[blue!45] plot[domain=-1:1,variable=\t,samples=31] (\xx,{\b*(1-\t*\t)},\t); }
  \draw[blue!45,dashed] (0,\b,0) -- (8,\b,0);
  \foreach \e/\c in \ends { \draw[blue!60!black,thick] plot[domain=1:-1,variable=\t,samples=31]
                             ({4+(\e-4)*(1-\t)/2},{\b*(1-\t*\t)},\t); }
  \node[blue!55!black,anchor=west] at (8.1,\b,0.15) {$\mathcal F_{72}^-$};
  \fill[red!10,opacity=0.75] plot[domain=1:-1,variable=\t,samples=31] (8,{-\b*(1-\t*\t)},\t)
                -- plot[domain=-1:1,variable=\t,samples=31] (0,{-\b*(1-\t*\t)},\t) -- cycle;
  \foreach \xx in {0,8} { \draw[red!55] plot[domain=-1:1,variable=\t,samples=31] (\xx,{-\b*(1-\t*\t)},\t); }
  \draw[red!55,dashed] (0,-\b,0) -- (8,-\b,0);
  \foreach \e/\c in \ends { \draw[red!70!black,thick] plot[domain=1:-1,variable=\t,samples=31]
                             ({4+(\e-4)*(1-\t)/2},{-\b*(1-\t*\t)},\t); }
  \node[red!60!black,anchor=east] at (-0.1,-\b,-0.35) {$\mathcal F_{72}^+$};
  \draw[densely dotted,very thick] plot[domain=0:0.97,variable=\t,samples=21] ({0.9+3.1*\t},{-\b*(1-\t*\t)},\t);
  \fill (0.9,-\b,0) circle (1.6pt);
  \node[anchor=north,align=center,font=\scriptsize,inner sep=1pt] at (1.25,-\b,-0.05) {four\\imaginary};
  \fill[orange!20,opacity=0.6] (-1.4,-1.8,1) -- (9.4,-1.8,1) -- (9.4,2.8,1) -- (-1.4,2.8,1) -- cycle;
  \draw[orange!65!black] (-1.4,-1.8,1) -- (9.4,-1.8,1) -- (9.4,2.8,1) -- (-1.4,2.8,1) -- cycle;
  \node[orange!50!black,anchor=south east] at (9.4,2.8,1) {$\mathcal F_{60}$};
  \draw[very thick] (0,0,1) -- (8,0,1);
  \node[anchor=west] at (8.05,0,1.06) {$\mathcal C$};
  \fill (4,0,1) circle (2pt) node[above=1pt] {$\ket\psi$};
  \foreach \e/\c in \ends { \fill (\e,0,-1) circle (1.6pt) node[below=1pt] {$\c$}; }
  \draw[gray] ([xshift=-4pt,yshift=-9pt]1.2,0,-1) -- ++(-0.8cm,-0.3cm) node[left,black,inner sep=1pt] {$\ket{\psi_-}$};
\end{tikzpicture}
\caption{The families of Proposition~\ref{prop:families}, with $\delta$ running vertically.  The faces $\mathcal F_{60}$
($\delta=1$) and $\mathcal F_{60}'$ ($\delta=-1$) are $11$-dimensional phase tori.  The components $\mathcal F_{72}^+$
(front) and $\mathcal F_{72}^-$ (back) are $10$-dimensional and complex conjugate to each other.  Both join the
$9$-torus $\mathcal C\subset\mathcal F_{60}$, on which the sums of $\alpha$ over $P_1$, $P_2$ and $P_3$ agree, to the
corresponding $9$-torus $\mathcal C'\subset\mathcal F_{60}'$.  The eight branches through $\ket\psi$, four in each
component, cross the slice $\delta=0$ of uniform amplitudes (dashed).  Up to diagonal local unitaries they end at four
real states, each reached by one branch from each component: $\ket{\psi_-}$ and the states obtained from it by negating
the coefficients on $P'_1$, $P'_2$ or $P'_3$ (numbers: their classes, Theorem~\ref{thm:classes}).  Dotted: a path in
$\mathcal F_{72}^+$ from a uniform state with four imaginary coefficients to $\ket\psi$.}
\label{fig:families}
\end{figure}

\begin{proposition}\label{prop:families}
(a) Up to diagonal local unitaries, the AME states with support exactly $E\cup P\cup P'$ are the states
$\ket{\Psi_\delta(\alpha,\beta)}$ with $|\delta|<1$ that satisfy Eq.~\eqref{eq:mixed}.  They form two components
$\mathcal F_{72}^\pm\ni\ket{\Psi_\delta(0,\pm\pi/2)}$, exchanged by complex conjugation, with $10$ parameters each.\\
(b) Up to diagonal local unitaries, the AME states with support exactly $E\cup P$ are the states
$\ket{\Psi_1(\alpha)}$, with $11$ parameters, and those with support exactly $E\cup P'$ are the states
$\ket{\Psi_{-1}(\beta)}$.\\
(c) $\ket{\Psi_1(\alpha)}$ is a limit of $\mathcal F_{72}^\pm$ if and only if the sums of $\alpha$ over $P_1$, $P_2$
and $P_3$ agree modulo $2\pi$, and $\ket{\Psi_{-1}(\beta)}$ is one if and only if the sums of $\beta$ over $P'_1$,
$P'_2$ and $P'_3$ agree.\\
(d) No state of $\mathcal F_{72}^\pm$ is real, even up to diagonal local unitaries.
\end{proposition}

\begin{proof}
By Lemma~\ref{lem:blocks}, everything reduces to linear congruences in the phases.  The ranks and invariant factors
quoted below are computed exactly in Ref.~\cite{repo}.  We use two facts.  (i) If a system of integer linear
congruences in $n$ phases has rank $\nu$ and invariant factors $d_1,\dots,d_\nu$ (the non-zero entries of the Smith
normal form of its matrix), its solutions, if any, form $d_1\cdots d_\nu$ disjoint $(n-\nu)$-dimensional tori.
(ii) In each reshaping, a diagonal local unitary adds to $\varphi_x$ a function of the row of $x$ plus a function of
its column, which cancels in every congruence~\eqref{eq:congruence}.

To tell the two components apart we use
\[
\omega=\sum_{x\in P',\,x_1\neq0}\varphi_x-\sum_{x\in P,\,x_1\neq0}\varphi_x .
\]
The nine mixed blocks without a state of first index $0$ contain each of these $18$ states exactly twice, so adding
their congruences in the form~\eqref{eq:mixed} gives $2\omega\equiv\pi$: every AME state with support $E\cup P\cup P'$
has $\omega\equiv\pm\pi/2$.  Diagonal local unitaries form a connected group that maps such states to such states, so
they cannot change this discrete value, while complex conjugation reverses it.

(a) The moduli equations of the $E$-blocks have rank $48$, so $s=1$ on $E$.  On a mixed block they give $s_x=s_y$,
$s_{x'}=s_{y'}$ and $s_x+s_{x'}=2$.  Every pair within a quartet lies in a mixed block, so $s$ is constant on each
quartet, and, as noted before Eq.~\eqref{eq:partition}, every $P_k$ shares a mixed block with every $P'_l$, so $s(P_k)+s(P'_l)=2$.
Hence $s=1+\delta$ on $P$ and $1-\delta$ on $P'$, with $|\delta|<1$ as $s>0$.  The $36$ congruences of the
$E$-blocks have rank $28$ in $48$ phases and invariant factors $1$, so by (i) their solutions form one $20$-dimensional
torus.  By (ii) it contains the phases of $\ket\psi$ shifted by any diagonal local unitary, which also form a
$20$-dimensional torus, so the two coincide.  Hence a diagonal local unitary brings every AME state with this support to some
$\ket{\Psi_\delta(\alpha,\beta)}$, and such a state is AME if and only if Eq.~\eqref{eq:mixed} holds.  This system has
rank $14$ in $24$ phases and invariant factors $1,\dots,1,2$, so by (i) its solutions form two disjoint
$10$-dimensional tori.  At $(\alpha,\beta)=(0,\pm\pi/2)$, $\omega=\pm9\pi/2\equiv\pm\pi/2$, so $\omega$ tells the two
tori apart.  Together with $\delta$ they give $\mathcal F_{72}^\pm$, on which $\omega\equiv\pm\pi/2$.  Complex
conjugation, $(\alpha,\beta)\mapsto(-\alpha,-\beta)$, exchanges them, and no diagonal local unitary does.  A diagonal
local unitary that maps a state of the family to another must fix the coefficients on $E$, and on $P\cup P'$ it then acts
by a common phase, possibly combined with the sign change of $P_3\cup P'$.  The common phase removes one of the $10$
phase dimensions, leaving $9$ parameters besides $\delta$.

(b) On $E\cup P$ the $1\times1$ blocks give $s=2$ on $P$ and no phase condition, so the states are
$\ket{\Psi_1(\alpha)}$ for all $\alpha$, and the common phase leaves $11$ parameters.  The same holds on $E\cup P'$.

(c) Let $A_k$ and $B_l$ be the sums of $\alpha$ over $P_k$ and of $\beta$ over $P'_l$.  Adding the two
congruences~\eqref{eq:partition} gives $A_k\equiv B_l$, so $A_1\equiv A_2\equiv A_3$ on $\mathcal F_{72}^\pm$, and also
on its limits, since $\ket{\Psi_{\delta_n}(\alpha_n,\beta_n)}\to\ket{\Psi_1(\alpha)}$ forces $\alpha_n\to\alpha$.
Conversely, let $A_1\equiv A_2\equiv A_3$ and $P'_l=\{w,x,y,z\}$.  Each of its six pairs lies in one mixed block, so
Eq.~\eqref{eq:mixed} prescribes the six pair sums $\sigma_{wx}\equiv\beta_w+\beta_x$, and so on.  By
Eq.~\eqref{eq:partition}, the sums of complementary pairs, $\sigma_{wx}+\sigma_{yz}$, $\sigma_{wy}+\sigma_{xz}$ and $\sigma_{wz}+\sigma_{xy}$, are
$A_1$, $A_2$, $A_3$ in some order, hence equal.  Therefore the solutions on $P'_l$ are exactly
$2\beta_w\equiv \sigma_{wx}+\sigma_{wy}-\sigma_{xy}$ and $\beta_x\equiv \sigma_{wx}-\beta_w$, $\beta_y\equiv \sigma_{wy}-\beta_w$,
$\beta_z\equiv \sigma_{wz}-\beta_w$: two solutions, differing by $\pi$ on all of $P'_l$, and eight in total.  Adding $\pi$ on
one quartet changes $\omega$ by $3\pi$, so four of the eight lie in each component, and for each of them
$\ket{\Psi_\delta(\alpha,\beta)}\to\ket{\Psi_1(\alpha)}$ as $\delta\to1$.  At $\delta=-1$ the same holds with $\alpha$
and $\beta$ exchanged.

(d) If $D\ket\Phi$ were real for some $\ket\Phi$ in $\mathcal F_{72}^\pm$ and a diagonal local unitary $D$, it would
have $\omega\in\{0,\pi\}$.  But it is AME with support $E\cup P\cup P'$, so $\omega\equiv\pm\pi/2$.
\end{proof}

\emph{Branches through the real state.}  At $\alpha=0$ all pair sums in the proof of (c) are $\sigma\equiv\pi$, so $\beta$
is constant on each quartet $P'_l$ with $2\beta\equiv\pi$, i.e.\ $\beta=\eta_l\pi/2$ with $\eta_l=\pm1$.  The eight
straight lines $\delta\mapsto\ket{\Psi_\delta(0,\beta)}$ are the \emph{branches} through $\ket\psi=\ket{\Psi_1(0)}$.
Since $\omega=3(\eta_1+\eta_2+\eta_3)\pi/2\equiv\eta_1\eta_2\eta_3\pi/2$, four lie in $\mathcal F_{72}^+$ and four in
$\mathcal F_{72}^-$.  Along a branch only the moduli change: as $\delta$ decreases from $1$ to $-1$, the coefficients on
$P$ shrink from $1/6$ to $0$ and those on $P'$ grow from $0$ to $1/6$.  At $\delta=-1$ a branch ends at
$\ket{\Psi_{-1}(\beta)}$ on the opposite face.  A common phase on $P'$, which a diagonal local unitary applies without
changing the coefficients on $E$, makes this state real, and for constant $\beta$ it becomes $\ket{\psi_-}$.  Up to diagonal
local unitaries, the eight branches thus end at four real states, each reached by one branch from each component:
$\ket{\psi_-}$ and the three states obtained from it by negating the coefficients on $P'_1$, $P'_2$ or $P'_3$ (classes
$9$, $9$, $10$, $9$ of Theorem~\ref{thm:classes}).

\emph{Uniform states.}  Halfway along a branch, at $\delta=0$, all $72$ coefficients have modulus $1/\sqrt{72}$.  For
constant $\beta=\pm\pi/2$ they are $\varepsilon_x$, $1$ and $\pm i$ times $1/\sqrt{72}$ on $E$, $P$ and $P'$.  In every
reshaping $W=\sqrt2\,U$ then has two entries from $\{\pm1,\pm i\}$ in each row and column and $WW^\dagger=2\,\mathbb I$:
a 2-unitary complex weighing matrix (Fig.~\ref{fig:amplitudes72}).  These uniform states are new: the AME$(4,6)$
states reported earlier have at least $112$ non-zero coefficients, and the real states of Refs.~\cite{Bru26,HZ26} have
$60$ and $96$ (Sec.~\ref{sec:discussion}).  Other uniform states of
$\mathcal F_{72}^+$ are almost real: up to a diagonal local unitary, $\mathcal F_{72}^+$ contains a uniform state with
coefficients $\pm1/\sqrt{72}$ except for $\pm i/\sqrt{72}$ on the four states of $P'_1$, and a path of AME states in
$\mathcal F_{72}^+$ joins it to $\ket\psi$ (Fig.~\ref{fig:families})~\cite{repo}.  Four imaginary coefficients are the
minimum.  If all coefficients of an AME state with support $E\cup P\cup P'$ lie in $\{\pm1,\pm i\}/\sqrt{72}$, then by
Eq.~\eqref{eq:congruence} every block contains an even number of imaginary coefficients, while $\omega\equiv\pm\pi/2$
requires an odd number among the $18$ states of $P\cup P'$ with $x_1\neq0$.  The smallest sets with both properties are
the six quartets~\cite{repo}.

The first-order deformations of an AME state in the complex state space, modulo the tangent space of its local-unitary orbit, form a
local-unitary invariant, in the spirit of the defect of a unitary matrix~\cite{TZ2008}.  Their dimension is $10$ for a generic member of $\mathcal F_{72}^\pm$, $11$ for a generic face
state, and $14$ at $\ket\psi$: $11$ directions along the face and $3$ spanned by the eight branches (numerical ranks,
with gaps of $10^{-15}$ against $10^{-2}$).  The real state is thus a singular point of the AME conditions, where the
$11$-dimensional face and both $10$-dimensional components meet (Fig.~\ref{fig:families}).
Since these deformations are complex, they say nothing about real ones; the real states are isolated
(Proposition~\ref{prop:real-isolation}).

\section{Classification of the real states}\label{sec:real}

\emph{Where the real states lie.}  None of the real states lies in $\mathcal F_{72}^\pm$
(Proposition~\ref{prop:families}(d)), so they all lie on the two faces.  Up to diagonal local unitaries, $\mathcal F_{60}$ is
the $11$-torus of the phases $\alpha$ modulo a common phase and the sign change of $P_3$, which a diagonal local
unitary with entries $\pm1$ performs without changing the coefficients on $E$.  Its origin is $\ket\psi$.  By
Lemma~\ref{lem:blocks}, a face state is real up to diagonal local unitaries if and only if its phases are $0$ or $\pi$
up to a diagonal local unitary, i.e.\ if and only if it is a point of order at most two of this torus.  The real states of
$\mathcal F_{60}$ are therefore exactly its $2^{11}=2048$ points of order at most two: $1024$ with all
$\alpha_x\in\{0,\pi\}$, and $1024$ with $\alpha_x=\pm\pi/2$ on $P_3$ and $\alpha_x\in\{0,\pi\}$ on $P_1\cup P_2$.  On
$\mathcal F_{60}'$, with origin $\ket{\psi_-}$, the phases $\beta$ are determined modulo a common phase alone, and its
$2048$ real states are those with all $\beta_x\in\{0,\pi\}$.  The closures of $\mathcal F_{72}^\pm$ meet
$\mathcal F_{60}$ only in the $9$-torus $\mathcal C$ on which the quartet sums of $\alpha$ agree
(Proposition~\ref{prop:families}(c)), and $\mathcal F_{60}'$ only in the corresponding $9$-torus $\mathcal C'$.  The $2^9=512$ points of order at most two of each are the real states that are
limits of $\mathcal F_{72}^\pm$.  The other $3072$ real states lie only on the faces.  The eight straight branches through a
real state of $\mathcal C$ end at real states of $\mathcal C'$ only for the $128$ states of class~$1$ below, the class
of $\ket\psi$ (at states of classes $9$ and $10$), and for the other $384$ they end at non-real states~\cite{repo}.

Equivalently, a real state of $\mathcal F_{60}$ or $\mathcal F_{60}'$ is given by signs that solve the $36$ congruences of
the $E$-blocks over $\mathbb F_2$: on each face there are $2^{32}$ such sign patterns, and modulo
diagonal local unitaries they give the $2048$ real states.  Negating a coefficient $1/6$ keeps a state AME, since these are $1\times1$ blocks.  On each
face, the permutations of the parties and of the basis states $\ket0,\dots,\ket5$ of each party that map the support onto itself, and each officer (one or two basis states) onto an officer, form a
group of order $192$, with $16$ orbits on the real states of $\mathcal F_{60}$ ($8$ of size $64$, $8$ of size $192$) and
$24$ on those of $\mathcal F_{60}'$ ($4$ of size $32$, $20$ of size $96$).

\begin{theorem}\label{thm:classes}
The $4096$ real AME$(4,6)$ states with support $E\cup P$ or $E\cup P'$ fall into exactly $20$ classes under local
unitaries and permutations of the parties: $8$ on $\mathcal F_{60}$ and $12$ on $\mathcal F_{60}'$.  Their
representatives are listed in Table~\ref{tab:classes}.
\end{theorem}

\begin{proof}
\emph{At most $20$.}  States in one orbit are related by monomial local unitaries and party permutations.  On each
face, orbits are further joined by the real local orthogonal map $H_{12}^{\otimes4}$, a $2\times2$ Hadamard rotation of
the basis states $\ket1,\ket2$ of every party, followed by a monomial map.  This leaves $8$ classes on $\mathcal F_{60}$ and
$12$ on $\mathcal F_{60}'$~\cite{repo}.\\
\emph{At least $20$.}  Local-unitary invariants have been used to separate AME states and their phase
families~\cite{RAL2022,RL2025}.  We use the spectrum of a Jacobian.  Let $J$ be the real Jacobian of the map $\Phi\mapsto(\rho_A)_{|A|=2}$ to all six two-party
marginals, and $n(t)$ the number of eigenvalues of $72J^TJ$ below $t$.  A local unitary acts orthogonally on the state
and on each marginal, and a party permutation permutes the marginals, so the spectrum of $J^TJ$, and hence $n(t)$, is
invariant under both.  The pairs $\bigl(n(149/8),n(3/2)\bigr)$ of the $20$ representatives are pairwise different
(Table~\ref{tab:classes}).  To certify them, note that $72J^TJ$ has entries in $\mathbb Z[\sqrt2]$ and splits into
blocks of size at most $160$.  For each block $G$ we form $X=V^T(G-t\mathbb I)V$ in interval (ball) arithmetic~\cite{Arb2017},
with $V$ the numerical eigenvectors of $G$, and check that every Gershgorin interval of $X$ excludes $0$.  By
Sylvester's law of inertia, the number of negative intervals is then the number of eigenvalues of $G$ below $t$.
\end{proof}

\begin{table*}[t]
\caption{The $20$ classes of Theorem~\ref{thm:classes}.  Face $\mathcal F_{60}$: $\ket\psi$ with the listed coefficients $1/6$
negated.  Face $\mathcal F_{60}'$: the same for $\ket{\psi_-}$, Eq.~\eqref{eq:psiminus}.  Limit: limits of $\mathcal F_{72}^\pm$.
Dim: first-order complex deformations modulo local unitaries (numerical, Sec.~\ref{sec:families}).
$n(t)$: the certified invariants of the proof.}
\label{tab:classes}
\begin{ruledtabular}
\begin{tabular}{clcccccc}
class & face & flipped & states & limit & dim & $n(149/8)$ & $n(3/2)$\\
\hline
1 & $\mathcal F_{60}$ & -- & $128$ & yes & $14$ & $1048$ & $157$\\
2 & $\mathcal F_{60}$ & $0344$, $3505$ & $384$ & yes & $14$ & $1046$ & $157$\\
3 & $\mathcal F_{60}$ & $5303$ & $128$ & no & $11$ & $1050$ & $153$\\
4 & $\mathcal F_{60}$ & $5430$ & $128$ & no & $11$ & $1044$ & $152$\\
5 & $\mathcal F_{60}$ & $5045$ & $128$ & no & $11$ & $1053$ & $153$\\
6 & $\mathcal F_{60}$ & $0533$ & $384$ & no & $11$ & $1055$ & $153$\\
7 & $\mathcal F_{60}$ & $0455$ & $384$ & no & $11$ & $1052$ & $153$\\
8 & $\mathcal F_{60}$ & $0344$ & $384$ & no & $11$ & $1052$ & $152$\\
9 & $\mathcal F_{60}'$ & -- & $288$ & yes & $14$ & $1021$ & $157$\\
10 & $\mathcal F_{60}'$ & $0434$, $3055$ & $96$ & yes & $14$ & $1029$ & $173$\\
11 & $\mathcal F_{60}'$ & $0434$, $5340$ & $96$ & yes & $14$ & $1025$ & $159$\\
12 & $\mathcal F_{60}'$ & $0353$, $4044$ & $32$ & yes & $14$ & $985$ & $155$\\
13 & $\mathcal F_{60}'$ & $0353$ & $288$ & no & $11$ & $1004$ & $152$\\
14 & $\mathcal F_{60}'$ & $0434$ & $288$ & no & $11$ & $1032$ & $160$\\
15 & $\mathcal F_{60}'$ & $0545$ & $288$ & no & $11$ & $1022$ & $152$\\
16 & $\mathcal F_{60}'$ & $3055$ & $288$ & no & $11$ & $1026$ & $152$\\
17 & $\mathcal F_{60}'$ & $0434$, $5405$ & $96$ & no & $11$ & $1038$ & $168$\\
18 & $\mathcal F_{60}'$ & $5033$ & $96$ & no & $11$ & $1020$ & $152$\\
19 & $\mathcal F_{60}'$ & $5340$ & $96$ & no & $11$ & $1012$ & $152$\\
20 & $\mathcal F_{60}'$ & $5405$ & $96$ & no & $11$ & $1030$ & $152$\\
\end{tabular}
\end{ruledtabular}
\end{table*}

Simpler invariants do not separate the classes: the power traces $\mathrm{Tr}(J^TJ)^k$, $k\le8$, are equal on
classes $3$--$8$ (numerical).
\begin{proposition}\label{prop:real-isolation}
Every real AME$(4,6)$ state $\ket\Phi$ with support $E\cup P$ or $E\cup P'$ is isolated among real AME states modulo
local orthogonal transformations: every real AME state close enough to $\ket\Phi$ is $(O_1\otimes O_2\otimes
O_3\otimes O_4)\ket\Phi$ for some $O_j\in O(6)$.
\end{proposition}

\begin{proof}
\emph{Reduction to $512$ states.}  The count of $4096$ real states is up to diagonal local unitaries, which may be
complex.  Here only real orthogonal maps are allowed, so we count the real states again.  On each face the
sign patterns of the real states form an affine space of dimension $32$ over $\mathbb F_2$, and two of them are related
by a diagonal local unitary if and only if they differ by an element of a $21$-dimensional subspace.  The sign changes
of single basis states of single parties, the real diagonal local unitaries, span only an $18$-dimensional part of it on
$\mathcal F_{60}$ and a $17$-dimensional part on $\mathcal F_{60}'$, so each of the $2048$ points of a face splits into
$2^3=8$, respectively $2^4=16$, classes under local sign changes.  We take one sign pattern from each of the $16$ and $24$
orbits of the support symmetry group of order $192$ and add the $8$, respectively $16$, coset representatives.  This
gives $16\cdot8+24\cdot16=512$ real states, and every real state on either face is mapped to one of them by local sign
changes, permutations of the basis states of each party and permutations of the parties~\cite{repo}.  These maps are
orthogonal, so it suffices to prove the claim for the $512$ states.

\emph{Linearization.}  Let $v=\sqrt{72}\,\ket\Phi\in\mathbb R^{1296}$, so that $v=a+\sqrt2\,b$ with integer vectors $a$
and $b$ supported on $E$ and on the classical officers.  Let $F(v)$ collect the upper triangles of $V V^T$ for the three
reshapings $V$ of $v$, $1998$ entries in all; then $\ket\Phi$ is AME if and only if $F(v)=F_0$, with $F_0$ the upper
triangles of $2\,\mathbb I_{36}$.  Since $F$ is quadratic, $F(v+h)=F(v)+Jh+F(h)$ exactly, with $J$ the derivative of
$F$ at $v$, a $1998\times1296$ matrix with entries in $\mathbb Z[\sqrt2]$.  The generators of $O(6)^{\times4}$ give $60$
vectors $T$ with $JT=0$, as identities in $\mathbb Z[\sqrt2]$.  To obtain ranks we reduce modulo the primes
$p=2^{31}-1$ and $p=2^{31}-105$, both $\equiv7\pmod 8$, so that $2$ is a square modulo $p$ and $\mathbb Z[\sqrt2]$ maps
to $\mathbb F_p$.  A rank modulo $p$ is a lower bound for the rank over $\mathbb Q(\sqrt2)$, and exact kernel vectors give
upper bounds.  This shows that $T$ has rank $60$ at all $512$ states.  At $400$ of them ($112$ on $\mathcal F_{60}$, $288$
on $\mathcal F_{60}'$), $J$ has rank $1236=1296-60$, so $\ker J$ is spanned by $T$.  At the other $112$ ($16$ and $96$),
$J$ has rank $1234$, and $\ker J$ is spanned by $T$ and two integer vectors $w_1,w_2$ supported on the $12$ states of
$P'$, respectively $P$, that $\ket\Phi$ does not use.  For these we also verify that $F(w_1+w_2)=F(w_1)+F(w_2)$ and that
$[J\,|\,F(w_1)\,|\,F(w_2)]$ has rank $1236$, so the classes of $F(w_1)$ and $F(w_2)$ in the cokernel of $J$ are
linearly independent~\cite{repo}.

\emph{Isolation.}  Let $N\subset\mathbb R^{1296}$ be a complement of the span of $T$ that contains $w_1,w_2$ when they
exist.  The derivative of $(g,h)\mapsto g(v+h)$ on $O(6)^{\times4}\times N$ at $(\mathbb I,0)$ is onto, so every real
state near $v$ is $g(v+h)$ with $g$ near $\mathbb I$ and $h\in N$ small.  It is therefore enough to show that $h=0$ is
an isolated solution of $Jh+F(h)=0$ in $N$.  If $J$ is injective on $N$, this holds because $|Jh|\ge c|h|$ and
$|F(h)|\le C|h|^2$.  Otherwise, since the solutions form a real algebraic set, a non-isolated solution would give a real
analytic arc $h(t)=t^mu+O(t^{m+1})$ of solutions in $N$, with $m\ge1$ and $u\ne0$ (curve selection lemma~\cite{Milnor1968}).  The order $t^m$
gives $Ju=0$, so $u=c_1w_1+c_2w_2$.  Projecting onto the cokernel of $J$ removes $Jh$, and the order $t^{2m}$ gives
$c_1^2[F(w_1)]+c_2^2[F(w_2)]=0$.  By the independence above, $c_1=c_2=0$, a contradiction.
\end{proof}

Within the families of Sec.~\ref{sec:families} the real states are therefore isolated points of complex continua.

The same computation bounds the local stabilizers.  At each of the $512$ states, the real symmetric generators of
$\mathfrak u(6)^{\oplus4}$, multiplied by $i$, give $84$ vectors of rank $81$ (the identity of each party acts as
$i\ket\Phi$), certified as above.  Together with the $60$ real vectors $T$, the local-unitary orbit has dimension $141$,
and $140=\dim PU(6)^{\times4}$ projectively.  The stabilizer of $\ket\Phi$ in $PU(6)^{\times4}$ is therefore finite.
Unlike Proposition~\ref{prop:real-isolation}, this is a property of the local-unitary class, and so it holds for all
states of the $20$ classes.

The deformation dimensions also separate our states from the known solutions: $14$ and $11$ for the real classes and
$10$ for $\mathcal F_{72}^\pm$, against $4$ for the golden state of Ref.~\cite{RBB22} and $3$ for the 2-unitary Hadamard
matrix $H(0)$ of Ref.~\cite{BZ24}.  None of our states is locally equivalent to those (numerical ranks, with the same gaps as above).

\section{Discussion}\label{sec:discussion}

The real states have $60$ non-zero coefficients.  Until Refs.~\cite{Bru26,Rat26,HZ26} the sparsest
reported AME$(4,6)$ state was the golden state with $112$~\cite{RBB22,RRKL23}, whose representation of numerically minimized R\'enyi-2
decomposition entropy has $144$~\cite{Ram25}.  Of the real states found in parallel, that of Ref.~\cite{HZ26} has
$96$, and those of Refs.~\cite{Bru26,Rat26} also $60$.  The
argument of Sec.~\ref{sec:state} also bounds the support of any AME$(4,6)$ state from below.  A coefficient of modulus
$1/6$ exhausts the norm of its row and column in every reshaping, so the classical officers form a code of length $4$
and minimum distance $3$ over six symbols.  Such a code has at most $34$ codewords~\cite{Parker1975}: $36$ would be a pair of orthogonal
Latin squares of order six, and $35$ complete uniquely to $36$.  Every other row needs at least two entries, so the
support is at least $72-t\ge38$, where $t$ is the number of classical officers.  Conversely, a state with $s<72$ non-zero
coefficients has at least $72-s$ classical officers.  All $36$ officers can be entangled only from support $72$ on, as in the
uniform states of Sec.~\ref{sec:families}, and the $12$ classical officers of our $60$-term states are the minimum for their
support.  Whether the gap between $38$ and $60$ can be closed is open.

\emph{Other real solutions.}  Bruzda's real state~\cite{Bru26} is locally equivalent to $\ket\psi$.  Rather's two
$60$-term real 2-unitaries~\cite{Rat26} are mapped by permutations of the parties and of the basis states onto real
states of $\mathcal F_{60}$ and $\mathcal F_{60}'$ in classes $1$ and $9$, so they are locally equivalent to
$\ket\psi$ and $\ket{\psi_-}$.  Its variant with
one coefficient negated~\cite{HZ26} is $\ket\psi$ with the classical officer $0455$ negated, the representative of
class~$7$, and the one-parameter phase family of Ref.~\cite{HZ26} is a curve in $\mathcal F_{60}$.  The 96-term real
state of Ref.~\cite{HZ26}, invariant under a simultaneous cyclic shift of all four parties, is of a different kind.
The degree-8 invariant $K$ of Ref.~\cite{HZ26}, evaluated on the tensor $6\Phi_{ijkl}$, the entries of the 2-unitary of a
state $\ket\Phi$, is $108$ for it, while it is $72$
on classes $1$--$2$, $40$ on classes $3$--$8$ and at most $0$ on classes $9$--$20$, so that state is not equivalent
to any of our real states.  ($K$ takes only $10$ distinct values on our $20$ classes.)  Its deformation dimension
is $3$, while every state of our families has at least $10$, so it lies outside both families.  It also has no
classical officer in any local basis: its overlap with every product state is below $0.98/6$, so by the
argument above each of its local representatives has at least $72$ non-zero coefficients (both numerical).  The
statements about the states of Refs.~\cite{Bru26,Rat26,HZ26} themselves were checked with the tools of Ref.~\cite{repo} but are
not among its scripts, which cover our states and the values of $K$ on them.  Real
coefficients alone therefore do not produce classical officers.  In our states they come from the faces $\mathcal F_{60}$ and $\mathcal F_{60}'$.

Further questions remain.  Proposition~\ref{prop:real-isolation} concerns the real states on the two supports.  A
complex local unitary can map a real state to another real state without being orthogonal, so whether every real state
of the $20$ classes is isolated modulo $O(6)^{\times4}$ is open.
Real AME$(4,6)$ states outside these families exist~\cite{HZ26}.  Whether other real states
have classical officers is open.  Do our real states admit a form with all amplitudes of modulus $1/36$, i.e.\ a
real-equivalent 2-unitary Hadamard matrix?  Such a form would still have product states of overlap $1/6$, the local
images of the $12$ states of $P$.
The eigenspace of $u^{\otimes4}$ also contains AME states
outside the families: about one in eight search hits lies on a $3$-parameter family of complex AME states with $348$
of the $360$ coefficients non-zero and deformation dimension $3$, which we could not bring to a sparse or real form
(numerical).  Whether it is related to the Hadamard family of Ref.~\cite{BZ24} or to the real state of
Ref.~\cite{HZ26} is open.  Both $H(0)$ and that state also have deformation dimension $3$.  Finally, a real 2-unitary of
order $36$ is a real perfect tensor: a real dual-unitary gate, and a building block of holographic codes with real
amplitudes~\cite{PYHP2015,BKP2019}.

\begin{acknowledgments}
This work was supported by the QuantERA II Programme (VERIqTAS project), which has received funding from the European Union’s Horizon 2020 research and innovation programme under Grant Agreement No. 101017733. We also acknowledge funding from the European Union’s Horizon Europe research and innovation programme under Grant Agreement No. 101080086 (NeQST). We thank Wojciech Bruzda and Remigiusz Augusiak for many fruitful discussions and guidance.  We are grateful to
Wojciech Bruzda, and to Mykhailo Hontarenko and Karol \.Zyczkowski, for sharing their independent real solutions with
us before publication, which made the comparison in Sec.~\ref{sec:discussion} possible.
\end{acknowledgments}

\paragraph*{Data and code availability.}  The state, the families and the class representatives, scripts that verify
every statement about them, and scripts that reproduce the two steps of the search are available in
Ref.~\cite{repo}.

\paragraph*{Use of AI tools.}  Anthropic's Claude models were used as assistants throughout this
work: for writing and checking the code of Ref.~\cite{repo}, for generating the tables and figures from the data, and for drafting and editing the text.  Every statement in the paper is verified by those scripts, and the author has reviewed all content and takes full responsibility for it.

\twocolumngrid
\bibliography{references}
\end{document}